\documentclass{IEEEtran}
\usepackage{amsmath,amsfonts}
\usepackage{algorithm}
\usepackage{array}
\usepackage{caption}
\usepackage{textcomp}
\usepackage{stfloats}
\usepackage{url}
\usepackage{verbatim}
\usepackage{graphicx}
\usepackage{cite}
\usepackage{booktabs}
\usepackage{bm}
\usepackage{tabularx}
\usepackage{amsthm}
\usepackage{graphicx}
\usepackage{subfigure}
\usepackage{algorithm}
\usepackage[noend]{algpseudocode}

\newtheorem{theorem}{Theorem}
\newtheorem{example}{Example}
\newtheorem{definition}{Definition}

\newcommand{\stitle}[1]{\vspace{1ex} 

\noindent{\bf #1}}
\long\def\comment#1{}

\begin{document}

\title{Top-k Representative Search for Comparative\\ Tree Summarization}

\author{Yuqi Chen, Xin Huang, Bilian Chen

\thanks{Yuqi Chen and Xin Huang are with Hong Kong Baptist University, Hong Kong, China. E-mail:\{csyqchen, xinhuang\}@comp.hkbu.edu.hk}~
\\
~\thanks{Bilian Chen is with the Department of Automation,
Xiamen University, China. E-mail: blchen@xmu.edu.cn}}



\maketitle

\begin{abstract}
Data summarization aims at utilizing a small-scale summary to represent massive datasets as a whole, which is useful for visualization and information sipped generation. 
However, most existing studies of hierarchical summarization only work on \emph{one single tree} by selecting $k$ representative nodes, which neglects an important problem of comparative summarization on two trees. 
In this paper, 
given two trees with the same topology structure and different node weights, we aim at finding $k$ representative nodes, where $k_1$ nodes summarize the common relationship between them and $k_2$ nodes highlight significantly different sub-trees meanwhile satisfying $k_1+k_2=k$. 
To optimize summarization results, we introduce a scaling coefficient for balancing the summary view between two sub-trees in terms of similarity and difference. 
Additionally, we propose a novel definition based on the Hellinger distance to quantify the node distribution difference between the sub-trees. We present an greedy algorithm SVDT to find high-quality results with approximation guaranteed in an efficient way. 
Furthermore, we explore an extension of our comparative summarization to handle two trees with different structures. Extensive experiments demonstrate the effectiveness and efficiency of our SVDT algorithm against existing summarization competitors.
\end{abstract}


\section{Introduction}
Graphs composed of nodes and edges are often used to depict complex datasets~\cite{jin2001structure},\cite{koutra2011algorithms}. Graph visualization gives a simplified presentation with an intuitive overview for users. However, graphs can only enhance understanding if the size of summary answers is small enough in a manageable scale for human comprehension \cite{jing2011graphical}. 

In this paper, we study a novel problem of differential tree summarization to depict the similarities and differences between two trees, which aims to show two trees' homogeneity and heterogeneity using a small summary answer. 
Let us consider a real example of an academic conference SIGMOD with several research sessions and papers. Fig.~1(a) shows two instance trees $T_1$ and $T_2$ representing the SIGMOD'19 conference, where the node weight represents the number of downloads of a research paper,  at different periods of 2019-2020 and 2019-2024, respectively. $T_1$ and $T_2$ have the same structure, which has four research sessions of 
\emph{MH, IE, DI, ML}, and 12 research papers ($p_1$, $p_2$ $\ldots$, $p_{12}$), but having different node weights. Specifically, $p_1$\footnote{Prasaad et al., Concurrent prefix recovery ..., In SIGMOD, 2019.}, $p_2$\footnote{Zhang et al., Briskstream: Scaling data stream ..., In SIGMOD, 2019.}, $p_3$\footnote{Kim et al., Border-collie: a wait-free, read-optimal ..., In SIGMOD, 2019.} are three research papers belonging to the research session of \emph{MH} (modern hardware).
Note that the weights of four research session nodes are the average download of this session's papers, to fairly represent the hotness of each research session. 

\begin{figure}[t]
\centering
\vspace{-0.4cm}
\includegraphics[width=3.8in]{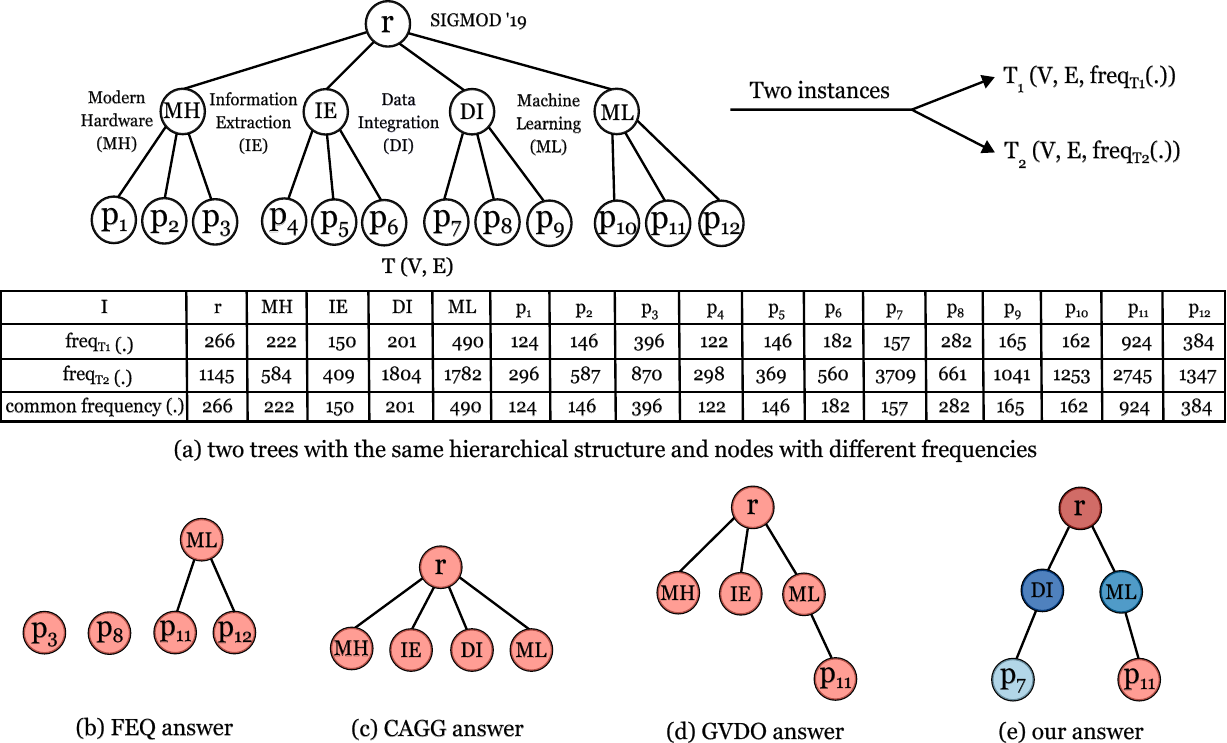}
\caption{A running example of our problem.} \label{fig1ex}
\vspace{-0.6cm}
\centering
\end{figure}

The objective of our differential summarization is to select the most $k$ \emph{attractive} or \emph{unpopular} research topics and papers to distinguish the conference SIGMOD'19 in the recent four years from 2020 to 2024. We set a small number $k=5$. The existing three different methods of FEQ~\cite{jing2011graphical}, CAGG~\cite{jing2011graphical}, and  GVDO~
\cite{zhu2021efficient} find the answers of five most attractive  papers, as shown in Figs.~1(b),~1(c),~1(d), respectively.  It lacks of identifying research sessions or specific papers that have gone from niche to popular, which fails to capture the changes in a summarized way. 
In contrast, our method can capture distinctive features between two trees with the same structure, as shown in Fig.~1(e). The answer includes $r$, \emph{DI}, \emph{ML}, $p_7$ and $p_{11}$. 
The nodes contributing to homogeneous summarization  are in red, while nodes contributing to differential summarization are in blue. The darker the node color, the more the node's contribution to the attribute it represents.



\begin{itemize}
    \item \textbf{Homogeneous summarization}: This involves two nodes $r$, $p_{11}$\footnote{Shang et al., Democratizing data science ..., In SIGMOD, 2019.} . Node $p_{11}$ is labeled in light red, indicating that this machine learning-related paper has been consistently popular from 2020 to 2024.
    \item \textbf{Differential summarization}: This involves three nodes \emph{DI}, \emph{ML}, $p_7$\footnote{Salimi et al., Interventional fairness: Causal ..., In SIGMOD, 2019.}. Both of the nodes \emph{DI} and \emph{ML} in dark blue with their direct ancestor node $r$ in dark red represent significant changes in data integration and machine learning, representing that these two sessions become increasingly popular from 2020 to 2024. The node $p_7$ is marked in light blue, meaning that from low at the beginning to very high later on, we find it is the best paper of SIGMOD'19. 
\end{itemize} 

In light of the above, the differential summary in Fig.~1(e) is a concise answer to well preserve the homogeneity and difference in $T_1$ and $T_2$. However, the problem of differential summarization has two challenges. First, it is hard to determine the number of homogeneous and differential nodes although the total number $k$ is given. Second, it is not easy to design a unified summary score function to similarity and difference simultaneously by incorporating the weight distributions in subtrees. To tackle them, we propose a fast  algorithm for differential summarization based on a new concept of distribution score $SimDif_D$ using Hellinger distance~\cite{hellinger1909neue}, which uses a scaling coefficient to balance. 
To summarize, this paper makes the following contributions:

\begin{itemize}
\item The problem we study is to choose a small set of elements to summarize two large-scale and hierarchically structured trees. We define the kVDT problem for finding representative nodes for similarity and difference summarization by considering 
the diverse representation across subtrees' weight distribution. 
(Section~3).
\item 
We propose an efficient Hellinger distance-based distribution calculation to quantify the  node distribution distance between two
sub-trees. Furthermore, we introduce a summary visualization technique that converts the output of our method into a concise and visually comprehensible tree summary (Section~5).
\item We conduct extensive experiments on real-world datasets to validate the feasibility and effectiveness of our proposed method against several competitors (Section~6). 
\end{itemize}

\vspace{-10pt}
\section{Related Work}

Graph summarization has been studied widely in the literature~\cite{navlakha2008graph,chu2024graph,kumar2018utility,wu2014graph,huang2017ontology}.
 Navlakha et al.\cite{navlakha2008graph} proposed a graph compression method to generate graph summary with edge corrections. To enhance efficiency, Chu et al. \cite{chu2024graph} improved Navlakha's approach by identifying and merging pairs with high saving values. 
Wu et al. \cite{wu2014graph} considered topology and feature similarity for summarizing attributed graphs, but led to a graph recovery error rate.
Kumar et al. \cite{kumar2018utility} and Lee et al. \cite{lee2022slugger} addressed this problem from different perspectives, the former focused on utility-driven to ensure the graph summary meet a user-specified utility threshold, while the latter achieved lossless summarization by representing unweighted graphs using positive and negative edges between hierarchical supernodes.
Huang et al. \cite{huang2017ontology} proposed GVDO algorithm to visualize a summary graph 
in the ontology structure. TS algorithm \cite{kim2020summarizing} was introduced for efficiently summarizing hierarchical multidimensional data. 
Different from most existing studies above that focus on a single graph, our study aims at differential summarization by finding $k$ similar and different representatives.  

\vspace{-10pt}

\section{Problem Statement}

In this section, we present our problem and notions.
\subsection{Preliminaries}
Given the tree $T = (\mathcal{V}, E, freq)$, which contains $\mathcal{V}$ nodes and $E$ edges, and $freq$ is a non-negative integer function of the node weight. We consider two instances of trees $T_1$ and $T_2$ with different weights.
Note that, in this paper, 
$anc(x)$ represents the ancestors of $x$ which include $x$ itself, while $des(x)$ represents $x$ itself with the descendants of $x$. 

To capture the representative effect of nodes and their descendants at different levels, we consider the distance denoted as $dis_x(y)$ between them. 
$dis_x(y)$ only exists when $y \in des(x)$
. Otherwise, $dis_x(y)$ defaults to 0.
The distance $dis_x(y)$ can be calculated using the following formula.

\vspace{-10pt}

\begin{equation*}
dis_x(y) =
\left\{
\begin{aligned}
\frac{1}{level(y)-level(x)+1}, if ~ y \in des(x)\\
0, otherwise
\end{aligned}
\right.
\end{equation*}

\begin{definition} [Differential Weight] Given two weighted trees $T_1, T_2$, the differential weight of vertex $x$ represented by $\omega(x)$ is the larger frequency in two trees, i.e.,  $\omega(x) = max(freq_{T_1}(x),freq_{T_2}(x))$. 
\end{definition}

    
To filter nodes with a larger frequency share in the original tree structure, we introduce the concept of differential weight in Definition 1. 

\begin{definition} [Scaling Coefficient] Given two weighted trees $T_1, T_2$, the scaling coefficient represented by $\gamma$ is obtained by calculating the average ratio of similarity and difference of all nodes satisfying $x \in \mathcal{V}$.
\begin{equation}
 \gamma = \frac{1}{|\mathcal{V}|} \sum_{x\in \mathcal{V}} \frac{\min(freq_{T_1}(x),freq_{T_2}(x))}{|freq_{T_1}(x)-freq_{T_2}(x)|}
 \label{eq.gamma}
\end{equation}
\end{definition} 

Note that Eq.~\ref{eq.gamma} ignores the case of node $x$ with  $freq_{T_1}(x)=freq_{T_2}(x)$. Similar computations are also applied to other metrics throughout this paper. In practical applications, there are significant differences in homogeneous and heterogeneous summarization due to varying frequencies and hierarchical structures of datasets. To balance the representation of similarity and difference in summary views, we introduce the scaling coefficient $\gamma$.


\vspace{-5pt}

\subsection{Summary Score Function}

\subsubsection{Distribution} 
In the hierarchical structure, we are concerned with the frequency allocation of themselves and their descendants. We measure node distribution by converting it into discrete probability distribution through normalization. 

\subsubsection{Self Feature} For a given node $x$, we define the measure formulas for calculating the similarity and difference in $T_1$ and $T_2$, expressed by $sim(x)$ and $dif(x)$, respectively. Note that the following definition only uses the scaling coefficient for the difference measure formula, since $\gamma$ is used to adjust the imbalance of the similarity and difference between nodes.

\vspace{-10pt}

\begin{equation*}
    sim(x)=\sum_{\substack{y\in des(x),\\ y\notin des(z),\\ z\in S_1}}{\frac{\min(freq_{T_1}(y),freq_{T_2}(y))}{\max(freq_{T_1}(y),freq_{T_2}(y))}\cdot dis_x(y)}
\end{equation*}

\vspace{-10pt}

\begin{equation*}
    dif(x)=\gamma \cdot \sum_{\substack{y\in des(x),\\ y\notin des(z),\\ z\in S_2}}\frac{|freq_{T_1}(y)-freq_{T_2}(y)|}{\max(freq_{T_1}(y),freq_{T_2}(y))}\cdot dis_x(y)
\end{equation*}

\begin{figure*}[t]
\vspace{-10 pt}
\centering
\includegraphics[width=\textwidth]{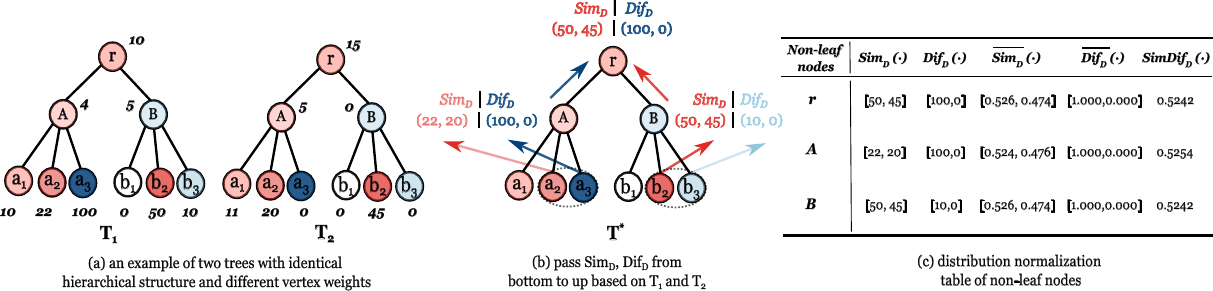}
\caption{An example of passing up the similarity and difference distribution. Here $\beta=1$. Finally, the root $r$ has 
$Sim_D(r)=[50, 45]$ and $Dif_D(r)=[100, 0]$. }
\label{fig1}
\vspace{-18 pt}
\centering
\end{figure*}

\subsubsection{Selective Gain Score} While a node is selected into similar set $S_1$ or differential set $S_2$, its selective gain score for $S_1$ is represented by $sim_{S_1}(x)$, and its selective gain score for $S_2$ is denoted as $dif_{S_2}(x)$, which are jointly determined by the node's weight, distribution and self feature. Here we use Hellinger distance to measure the distribution distance denoted as ${SimDif_D(x)}$ of the node $x$ between two trees $T_1, T_2$.
\begin{equation*}
    sim_{S_1}(x) = \omega(x) \cdot (1-SimDif_D(x)) \cdot sim(x)
\end{equation*}
\begin{equation*}
    dif_{S_2}(x) = \omega(x) \cdot SimDif_D(x) \cdot dif(x)
\end{equation*}
\subsubsection{Summary Score} Based on the definition of gain score, the summary score of the set $S$ is defined as:
\begin{equation*}
sum(S) = \sum \max_{x \in S_1 \cap anc(y) } sim_{S_1}(y) + \sum \max_{x \in S_2 \cap anc(y) } dif_{S_2}(y)
\end{equation*}

To summarize, this paper addresses the kVDT problem, which selects \textbf{k} representative nodes for \textbf{V}isualizing similarity and \textbf{D}ifference on hierarchical \textbf{T}ree structure as follows.

\subsubsection{kVDT Problem} Given two trees $T_1, T_2$ with different weights, and a positive integer $k$, find a representative set satisfying $S \subseteq \mathcal{V} $ and $S$ = $S_1 \cup S_2$, where $S_1$ includes  $k_1$ representatives of similarity with $S_2$ includes $k_2$ representatives of difference, ensuring $S$ reaches the maximum summary $sum(S)$ with $|S|=k$.

\section{Problem Analysis}
In this section, we analyze the property of our summary score function.

\stitle{Submodularity.} Given a set $S_i$, when $x\in S_i$, $Max_{S_i}(y)$ represents the maximum gain score for $y$ in the set $S_i$. In addition, we define $Rep_x(y)$ represents the representative effect of the as: $Rep_x(y) = \{y\in des(x)| freq(y) \cdot dis_x(y)\}$, and $C_{S_i}(x)$ indicates $x$ is the nearest ancestor of all nodes in set $C_{S_i}(x)$ in $S$ as: $C_{S_i}(x) = \{y\in des(x)|Max_{S_i}(y)=Rep_x(y)\}$.

\begin{theorem}
    $sum$ is submodular. i.e., for all $S_a$, $S_b$ $\subseteq \mathcal{V}$ subject to $S_a \subseteq S_b$, we have $sum(S_a \cup {x})-sum(S_a) \geq sum(S_b \cup {x})-sum(S_b).$
\end{theorem}
\begin{proof}
    Given two sets $S_1 \subset T_1 \subset \mathcal{V}$ for representing similarity and $S_2 \subset T_2 \subset \mathcal{V}$ for representing difference. Any element $x \in \mathcal{V} \setminus \{T_1\cup T_2\}$ can be added to these sets. Therefore, $T_i'= T_i\cup \{x\}$ and $S_i' = S_i \cup \{x\}.$
    Firstly, $\forall y \in V $, $Max_{S_i}(y) \leq Max_{T_i}(y)$ and $Max_{S_i'}(y) \leq Max_{T_i'}(y)$ is obvious. Then, $\forall y \in C_{T_i'}(x)$, we can figure out $Rep_x(y) = Max_{T_i'}(y) \geq Max_{S_i'}(y)$ from above. Besides, since $x$ belongs to set $S_i'$, $Rep_x(y) \leq Max_{S_i'}(y)$ holds. So $Rep_x(y) \leq Max_{S_i'}(y)$. According to the definition of function $C$, we can obtain $y\in C_{S_i'}(x)$. Combining the range of $y$, $C_{T_i'}(x)$ is a subset of $C_{S_i'}(x)$. Lastly, we denote summary score function $sum$. 
    Therefore, $sum(T_{i}')-sum(T_i)=\sum_{y\in V}(Max_{T_i'}(y)-Max_{T_i}(y))=\sum_{y\in C_{T_i'}(x)}(Rep_x(y)-Max_{T_i}(y))$. Similarly, we can obtain $sum(S_i')-sum(S_i)=\sum_{y\in C_{S_i'}(x)}(Rep_x(y)-Max_{S_i}(y))$. In view of $C_{T_i'}(x) \leq C_{S_i'}(x)$, we can figure out $\sum_{y\in C_{S_i'}(x)}(Rep_x(y)-Max_{S_i}(y)) \geq \sum_{y\in C_{T_i'}(x)}(Rep_x(y)-Max_{T_i}(y))$. That is to say, $sum(S_i')-sum(S_i)\geq sum(T_i')-sum(T_i)$. 
    
    $sum$ is a submodular function. Since submodular functions are additive and monotonic, $\sum_i sum(S_i')-\sum_i sum(S_i) \geq \sum_i sum(T_i')-\sum_i sum(T_i)$. This implies that both $S_1$ and $S_2$ satisfy submodularity. Therefore, the set $S$ obtained by merging $S_1$ and $S_2$ also satisfies submodularity.
\end{proof}

\stitle{Approximation}. In Theorem 1, we establish the submodularity of the summary score function 
$sum$. Thus, 
A greedy algorithm can be developed to  achieve $(1-1/e)$-approximation for maximizing a monotone submodular set function with a cardinality constraint~\cite{nemhauser1978analysis}. 


\section{SVDT Algorithm}
In this section, we introduce
a new greedy algorithm SVDT.

\vspace{-10pt}

\subsection{Distribution normalization of non-leaf nodes.}
In our problem formulation, we consider the distribution of nodes with their descendants. 
The key of comparing a representative node $x$ in two trees $T_1$, $T_2$ is to 
measure the similarity and difference of their representative normalized distribution in $x$'s subtree. 
Here, we use a parameter $\beta\in \mathcal{Z}^{+}$ to be the number of representative similar and different weights, 
i.e., $Sim_D(x) = [Sim_1 T_1(x), Sim_1 T_2(x), \ldots, 
Sim_\beta T_2(x)]$, $Dif_D(x) = [Dif_1 T_1(x), Dif_1 T_2(x), \ldots, 
Dif_\beta T_2(x)]$, respectively. 
These two distributions can be efficiently computed from the non-leaf nodes at the bottom level, i.e., progressively passing up top-$\beta$ similar and differential distributions from the bottom upwards until reaching the root. 


\begin{example}
    We use Fig.~2 to illustrate the distribution normalization of non-leaf nodes in detail. $T_1$ and $T_2$ share the same hierarchical structure and different weight values in Fig.~2(a). 
    Note that, we consider the common weight of nodes as the similarity value and the difference in weights as the difference value. For passing upward, we select top-1 node each from similarity distribution and difference distribution.
    Specifically, starting from the lowest level of non-leaf nodes, that is, node $A$, and considering it with its direct descendants, i.e., $\{A, a_1, a_2, a_3\}$, we select a node that best represents similarity, which is $a_2$, and a node that best represents difference, which is $a_3$, and use their frequencies in $T_1$ and $T_2$ as the distribution for the non-leaf node $A$. Accordingly, the distribution for the non-leaf node $B$ is represented by the frequencies of the most similar node $b_2$ and the most dissimilar node $b_3$ in $T_1$ and $T_2$. Therefore, we obtain $Sim_D(A)$, $Dif_D(A)$, $Sim_D(B)$, $Dif_D(B)$.
    Then, we proceed to pass upwards. For the root node $r$, with the nodes $A$ and $B$ as its direct descendants, we search for the distribution that best represents homogeneity and difference within $\{r, A, B\}$, which are obviously $Sim_D(B)$ and $Dif_D(A)$. Therefore, the distribution representing node $r$ is $Sim_D(r)$=(50,45), $Dif_D(r)$=(100,0).
    After obtaining $Sim_D$ and $Dif_D$ for all non-leaf nodes, as in Fig.~2(b), we proceed with the normalization to obtain $\overline{Sim_D}$ and $\overline{Dif_D}$ for all non-leaf nodes in Fig.~2(c). 
    After converting them to probability distributions, we apply the Hellinger distance to 
    calculate the distance between similarity distribution and difference distribution, i.e., ${SimDif_D}$ as the distribution values of the non-leaf nodes.
\end{example}

\begin{algorithm}[t]
\caption{Computing 
$SimDif_D(x)$}
\label{algo.1}
\begin{algorithmic}[1]
\footnotesize
\Require  A non-leaf node $x \in \mathcal{V}$, two weighted trees $T_1, T_2$. 
\Ensure $SimDif_D(x) \in [0,1]$.


\State 	Let $SimDif_D(x) \leftarrow 0$;

\State $ Sim_D(x) \leftarrow [Sim_1 T_1(x), Sim_1 T_2(x), 
...,  Sim_\beta T_1(x), 
Sim_\beta T_2(x)]$;
\State $ Dif_D(x) \leftarrow [Dif_1 T_1(x), Dif_1 T_2(x), 
..., Dif_\beta T_1(x), 
Dif_\beta T_2(x)]$;

\State $sum_{Sim_D}(x) \leftarrow \sum_{1\leq k \leq \beta, 1\leq n \leq 2} Sim_k T_n(x)$;
\State $sum_{Dif_D}(x) \leftarrow \sum_{1\leq k \leq \beta, 1\leq n \leq 2} Dif_k T_n(x)$;

\State $\overline{Sim_D(x)} \leftarrow  Sim_D(x)/ sum_{Sim_D}(x)$;
\State $\overline{Dif_D(x)} \leftarrow  Dif_D(x)/ sum_{Dif_D}(x)$;

\For {$i = 1$ to $\beta$}
\For {$j = 1$ to $2$}
\State  $SimDif_D(x) \leftarrow {SimDif_D}(x) + (\sqrt{Dif_i T_j(x)}-\sqrt{Sim_i T_j(x)})^2 $;
\EndFor
\EndFor
\State $SimDif_D(x) \leftarrow \frac{1}{\sqrt{2}} \sqrt{{SimDif_D(x)}} $;
\State \Return ${SimDif_D(x)}$;
\end{algorithmic}
\end{algorithm}

\stitle{Selective marginal gain.}  
For any summary set $ S \in \mathcal{V} $ and $ \forall \ x(x \in \mathcal{V} \ \land \  x \notin S $), we have $ \Delta sum (x|S) = sum(S\cup {x}) - sum(S) \geq 0 $. In our algorithm, although inserting a new node $x$ can produce a non-negative gain on the original set $S$, but this marginal gain is selective, we can choose to insert $x$ into $S_1$ for representing similarity or $S_2$ for representing difference according to its self feature score. 
Let $S= S_1 \cup S_2$. When $x$ is chosen as the similar marginal gain, $ \Delta sum
(x|S) = sim_{S_1}(S_1 \cup x) - sim_{S_1}(S_1)$. Conversely, if $x$ is a representative of difference, $ \Delta sum
(x|S) = dif_{S_2}(S_2 \cup x) - dif_{S_2}(S_2)$.


\stitle{Greedy marginal gain selection for fractional knapsack.} 
We reformulate our kVDT problem as the fractional knapsack problem. Given $n$ candidate nodes $x$ and an empty set $S$, where node $x$ has its total value $V(x)$ with different unit values as similarity and difference, denoted as $V_{sim}(x)$ and $V_{dif}(x)$ respectively. Each node weight is fixed to be 1, and the capacity of set $S$ is $k$. Assume that $V(l, S)$ is the total value of the set $S$ with $l$ elements. We have  $V(l, S)=V(l-1, S')+$ $\max_{x \in \mathcal{V}\setminus S'}(V_{sim}(x),V_{dif}(x)), \text{ where } 1\leq l\leq k$. 
For initialization, we set $V(l, S) = 0$ with $|S|=0$. We use SVDT in Algorithm 2 to solve the problem.


   

\stitle{Optimized summary by $k_1$ and $k_2$ combination.} To obtain the solution of kVDT problem, we propose an optimized summary combination. In each combination, we have a set $S$ with $|S| = k$ consisting of $k_1$ nodes representing similarity and $k_2$ nodes representing difference, where $k_1$ decreases from $k$ and $k_2$ increases from $0$ satisfying $k_1 + k_2 = k$. If $k_1 > k_2$ in the current combination, similarity is prioritized, otherwise, difference is prioritized. We calculate the changes in similarity and difference as $\Delta loss$ and $\Delta gain$ respectively. If $\Delta loss + \Delta gain > 0$, the elements in set $S$ are updated to the current combination. After the iteration, we obtain the optimized summary combination sets $S_1$ and $S_2$.






\begin{algorithm}[!t]
\caption{SVDT$(S_1,S_2,k_1, k_2, k)$}
\begin{algorithmic}[1]
\footnotesize
\Require A summary set $S$, two weighted trees $ T_1, T_2$, a positive integer $k$, a priority queue $CandidateByScore$.
\Ensure A set of $k_1$ similarity summary elements $S_1$, and a set of $k_2$ difference summary elements $S_2$.
\State 	Let $S \leftarrow \emptyset, S_1' \leftarrow \emptyset, S_2' \leftarrow \emptyset$, $CandidateByScore\leftarrow \emptyset$;
\State Compute initial score ${SimDif_D}$ for all non-leaf nodes satisfying $x\in \mathcal{V} \cap {des}(x)>1$  by Algorithm~\ref{algo.1} and push into $candidateByScore$;
\While{$|S| < k$}
\State $x \leftarrow CandidateByScore.top()$;
\If{$V_{sim}(x) > V_{dif}(x)$} 
\If{$ RoundSim[x] < |S_1|$}
    \State Update $V_{sim}(x)$;
$RoundSim[x] \leftarrow |S_1|$;
continue;
\Else \State {$S_1 \leftarrow  S_1 \cup \{x\}$;  }
\EndIf
\Else 
 \If{$ RoundDif[x] < |S_2|$}
     \State 
     Update $V_{dif}(x)$;
 $ RoundDif[x] \leftarrow |S_2|$;
continue;
 \Else
 \State {$S_2 \leftarrow  S_2 \cup \{x\}$;  }
 \EndIf
\EndIf
\State $S \leftarrow  S \cup \{x\}$;  
\EndWhile
\State \Return $S_1,S_2$;








\end{algorithmic}
\end{algorithm}

\vspace{-10pt}

\subsection{Summary visual representation.}

We make the summary visual representation based on SVDT answer $S$ as follows. Consider a hierarchical tree $T$ in Fig.~3(a) and $k = 7$.
Four nodes in the set $S_1$ contributing to similarity are in red color, while three nodes in the set $S_2$ contributing to difference are in blue color.
Actually, $T$ is an instance subtree of real ACM CCS hierarchical tree in Fig.~9(a). 
To achieve a compact visualization, we take two steps. Firstly, for nodes that are neither colored themselves nor have colored descendants, we prune them and their descendants. The remaining nodes form the topology structure shown in Fig.~3(b).
Secondly, we skip non-representative nodes in chains of representative nodes with ancestor-descendant relationships. For instance, in the chains \{$I_1, I_6, I_{13}$\} and \{$I_1, I_6, I_{14}$\}, we only show the ancestor-descendant relationships from $I_1$ to $I_{13}$ and from $I_1$ to $I_{14}$. In such cases, we represent their connection with dashed lines in the summary visual representation $T^*$ shown in Fig.~3(c). 

\begin{figure}[t]
\centering
\vspace{-5pt}
\includegraphics[width=3.5in]{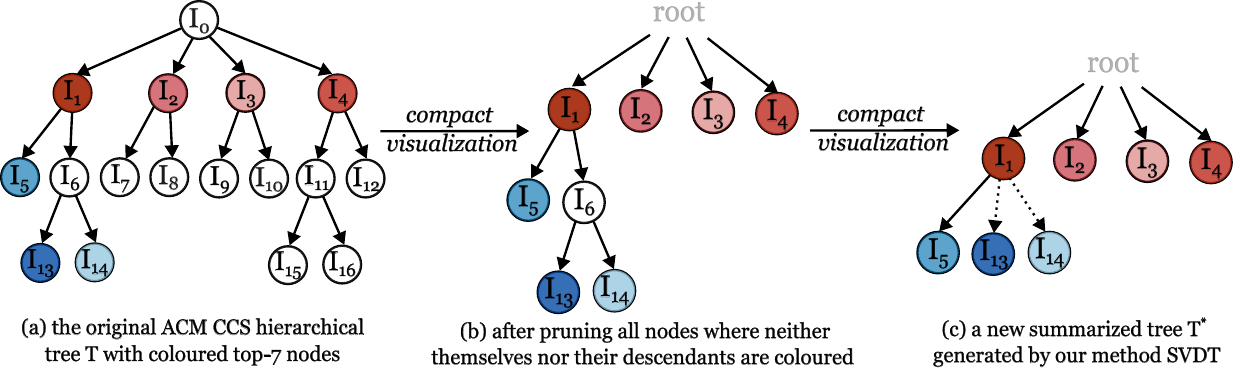}
\caption{Summary visualization  
based on SVDT answers.}
\label{fig1}
\vspace{-5pt}
\centering
\end{figure}

\vspace{-10pt}

\subsection{Discussion of handling two trees with different structures.} Since comparisons between trees with different structures are also widespread in practical situations, we discuss the solution of this situation in Fig.~4. Consider the relationship between two trees with different structures, 
it can be categorized into two cases, including (a) in which no subtree matching found and (b) in which one tree is a subtree of the other. In case (a), we compare the structural differences between the trees and introduce new edges and nodes at the corresponding positions. 
These newly introduced nodes are assigned a weight of 0 and transformed into the same structure for subsequent similarity and difference comparison. In case (b), where a subtree relation exists, if the size difference between $T_1$ and $T_2$ is relatively significant, extensive introduction operations in case (a) would be time-consuming. Instead, we efficiently compare the similarity and difference between such trees by subtree matching. Regarding this case, we consider the proportion of subtree matching, weights, and hierarchical differences to compute the commonality and diversity between the two trees. Specifically, in case (b), we consider the matching measure of $\{A, a_1, a_2\}$ in $T_2$ compared to $T_1$ in the above aspects.

\begin{figure}
[t]
\centering
\includegraphics[width=3.5in]{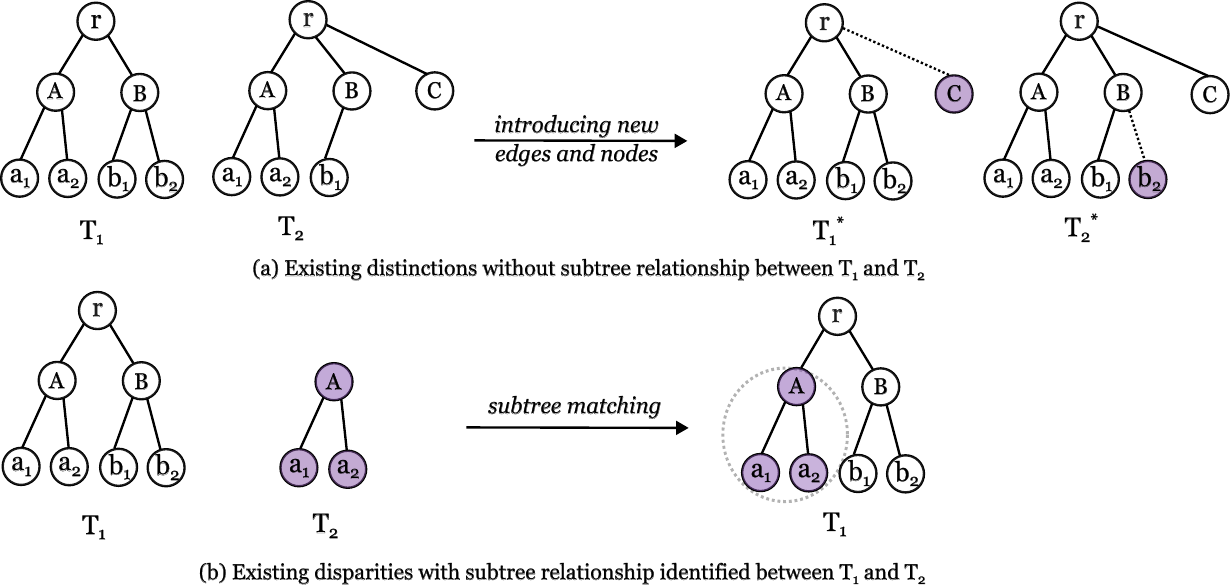}
\caption{Handle two trees' summarization in different structures.} \label{fig1}
\centering
\vspace{-15pt}
\end{figure}

\vspace{-5pt}

\section{Experiments}

In this section, we conducted extensive experiments to evaluate 
the performance. The source code of our algorithm is available at \url{https://github.com/csyqchen/TKDE_SVDT}.

\stitle{Datasets.} We used three pairs of real-world datasets containing hierarchical terminologies.  
Latt\&Lnur are extracted from the Medical Entity Dictionary \cite{jing2011graphical} with 4226 entities in a tree height of 22, which consists of information on patient and nurse access to online knowledge resources. 
The Anime hierarchical structure, extracted from the “Anime” directory in Wikipedia, includes 15,135 animes and their 
frequency represented by the number of page views in Jan 2018 and Jan 2024 \cite{pageviews}. 
The Yago hierarchical structure is an extracted tree of ontology structure yagoTaxonomy with 493,839 taxonomies  in multilingual Wikipedias, which are by randomly assigning two different weights denoted as Yago1 and Yago2.


\stitle{{Methods Comparison.}} To evaluate the quality of our method SVDT, we compare it with three baseline algorithms, including GVDO 
\cite{zhu2021efficient}, FEQ \cite{jing2011graphical}, and CAGG \cite{jing2011graphical}. GVDO is a summarization approach applied to ontology-based structure. FEQ is a method of selecting $k$ nodes with the highest frequency  as the summarized answer. CAGG selects $k$ nodes with the highest frequency as the summarization answer by controlling descendant contribution. 
Note that these three competitor methods are designed to work on \emph{a single tree}. For comparison, we implement variants of methods running on a \emph{a common tree}, which takes the minimum weight of a node in two trees $T_1$ and $T_2$ as the new weight in this common tree.  
In addition, we set the parameters $k=10$ and $\beta=50$ for SVDT by default.


\stitle{Evaluation Metrics.} To evaluate the quality of  summarization answers, we use three metrics: the diversity $Div(S)$, the query closeness $C_Q(S)$ \cite{huang2017ontology}, and average level difference $Ald(S)$ \cite{zhu2021efficient} as follows. 

First,  $Div(S)$  computes the diversity of the original data set represented by the summary node in two trees $T_1$, $T_2$. The higher the diversity, the better the summary.
\vspace{-3pt}
\begin{equation*}
    Div(S)=\sum_{\substack{
    x \in S \cap
    anc(y)}} |freq_{T_1}(y)-freq_{T_2}(y)|\cdot dis_x(y)
\end{equation*}
\vspace{-5pt}

Second, we design a query closeness $C_Q(S)$ to test the summary $S$'s distance to query $Q$.  
We randomly selected 500 vertices $Q$. For each vertex $q\in Q$, we accumulates the sum of the minimum distance from $q$ to $S$. 
A lower $C_Q( S)$ indicates a superior summary.

\vspace{-7pt}
\begin{equation*}
    C_Q(S) = \sum_{q \in Q} \min_{x \in S} \text{dist}_{T}(q, x)
\end{equation*}
\vspace{-7pt}


Third, the average level difference $Ald(S)$ is defined as the average level difference between summary node and weighted node, denoted by the following formula. 
A smaller average level difference means that the summary graph more effectively preserves the hierarchical structure of the original graph. 
\vspace{-3pt}
\begin{equation*}
    Ald(S) = \frac{\sum\limits_{y \in \mathcal{V}} \min\limits_{x \in S\cap anc(y)} 
    (level(y)-level(x))
    \cdot \alpha(x, y)
    }{\sum\limits_{y \in \mathcal{V}} \alpha(x, y)}
\end{equation*}, where $level(x)= dist_T(x, r_0)$ in tree $T$ rooted by $r_0$ and  $\alpha(x,y)=|freq_{T_1}(y)-freq_{T_2}(y)|$.   

\stitle{Exp-\MakeUppercase{\romannumeral 1}: Quality Evaluation}.  Fig.~5 demonstrates that SVDT markedly surpasses three baseline methods in capturing diversity $Div(S)$, thereby indicating its superiority in summarizing intrinsic tree relationships. Fig.~6 and Fig.~7 both illustrate that our algorithm, SVDT, maintains the closest distance and the smallest average level difference in almost all datasets compared to other baseline methods. This demonstrates that our summary is the most effective in closely mirroring the original tree structure of the dataset. 

\vspace{20pt}

\begin{figure}[h]
\vspace{-10pt}
\centering
\hspace{-1.5em}
\subfigure[Latt$\&$Lnur]{
\includegraphics[width=1.25in]{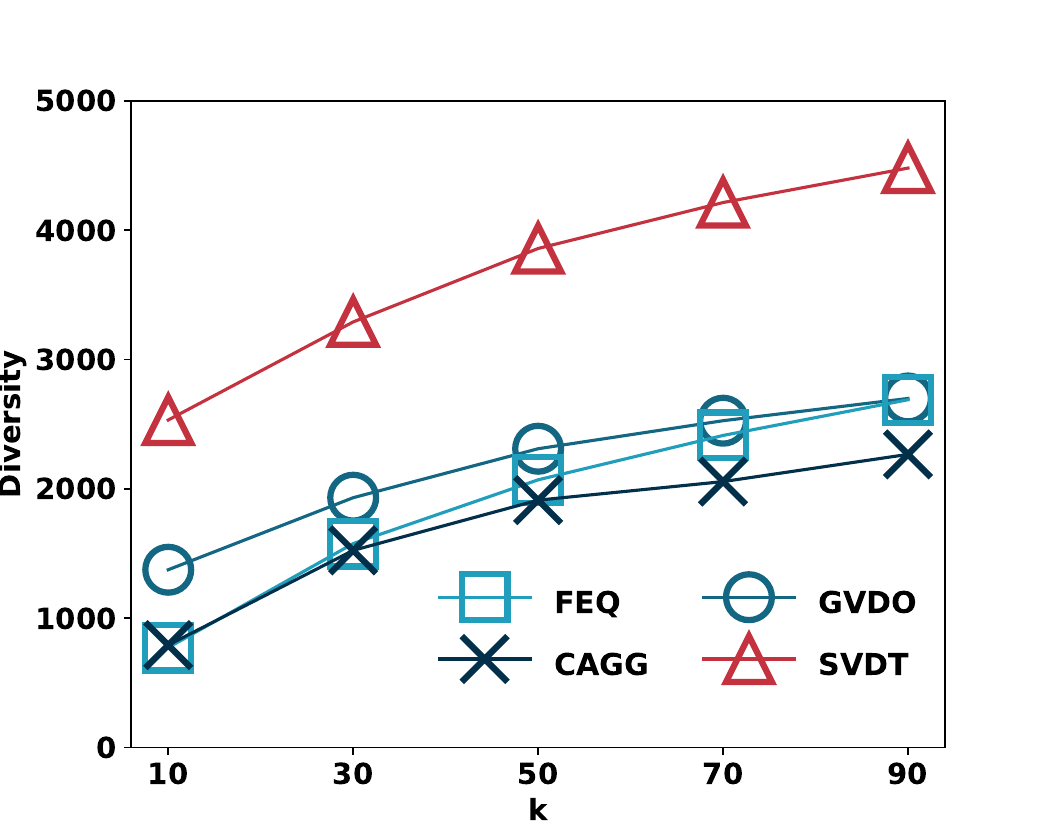}}
\hspace{-1.8em}
\subfigure[Anime2018$\&$2024]{
\includegraphics[width=1.25in]{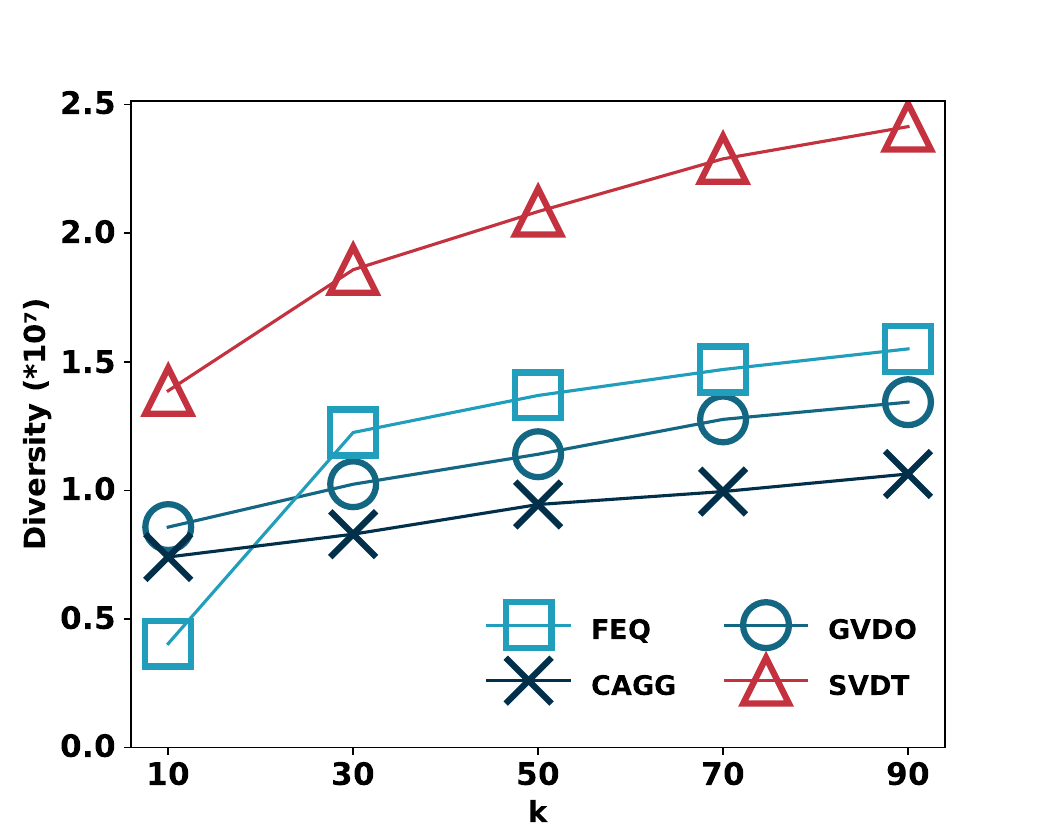}}
\hspace{-1.8em}
\subfigure[Yago1$\&$2]{
\includegraphics[width=1.25in]{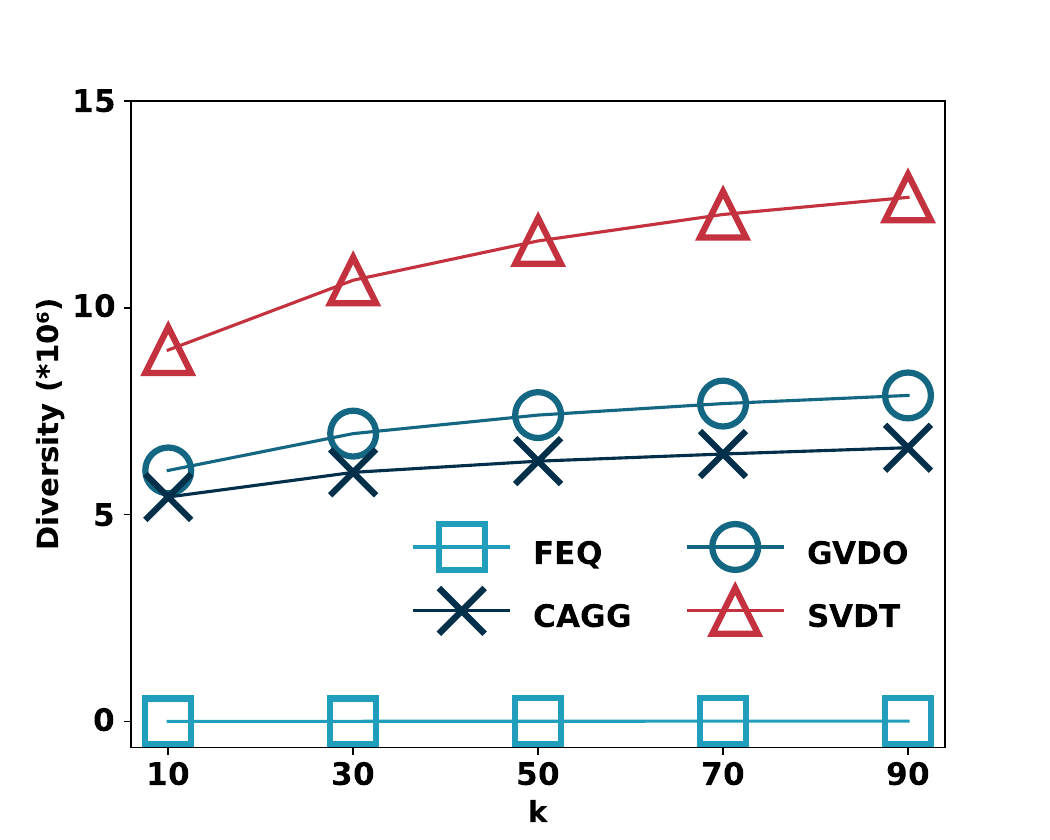}}
\hspace{-2.5em}
\vspace{-5pt}
\caption{Diversity $Div(S)$ evaluation on 
all datasets.}

\label{Fig.main}
\end{figure}


\begin{figure}[h]
\vspace{-15pt}
\centering
\hspace{-1.5em}
\subfigure[Latt$\&$Lnur]{
\includegraphics[width=1.25in]{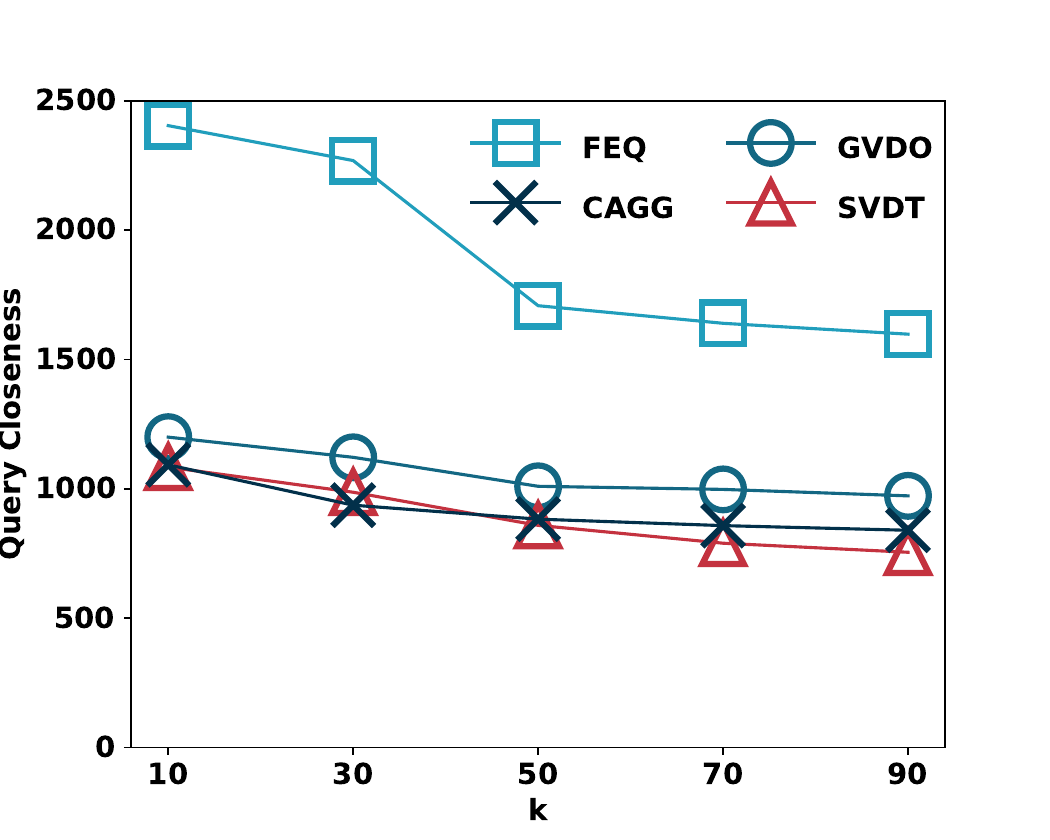}}
\hspace{-1.8em}
\subfigure[Anime2018$\&$2024]{
\includegraphics[width=1.25in]{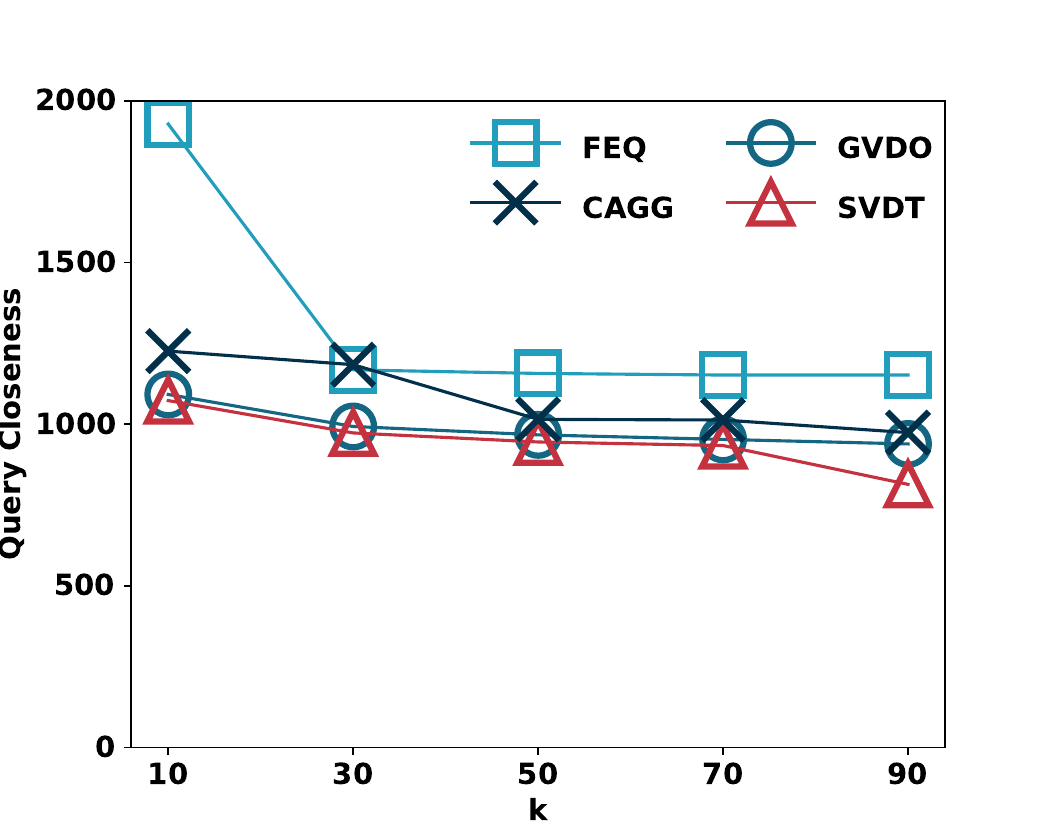}}
\hspace{-1.8em}
\subfigure[Yago1$\&$2]{
\includegraphics[width=1.25in]{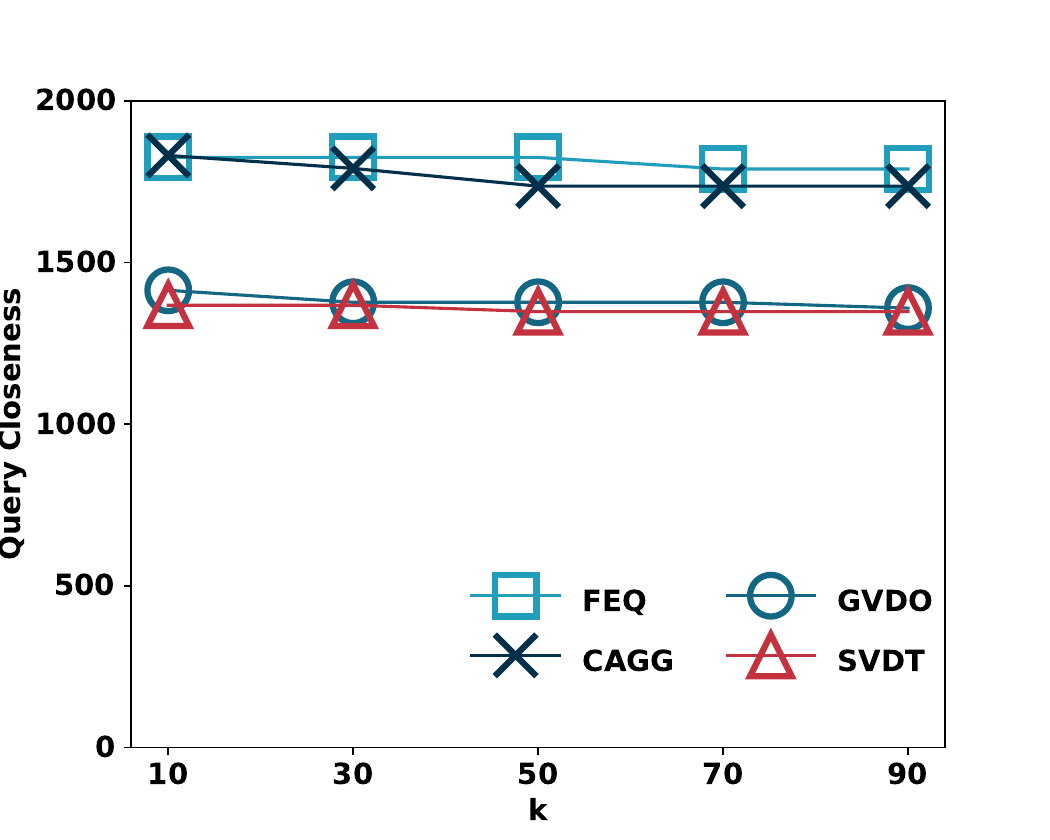}}
\hspace{-2.5em}
\vspace{-5pt}
\caption{Query closeness $C_Q(S)$ evaluation on 
all datasets.}
\label{Fig.main}
\end{figure}
\vspace{-23pt}

\begin{figure}[!h]
\vspace{-15pt}
\centering
\hspace{-1.5em}
\subfigure[Latt$\&$Lnur]{
\includegraphics[width=1.25in]{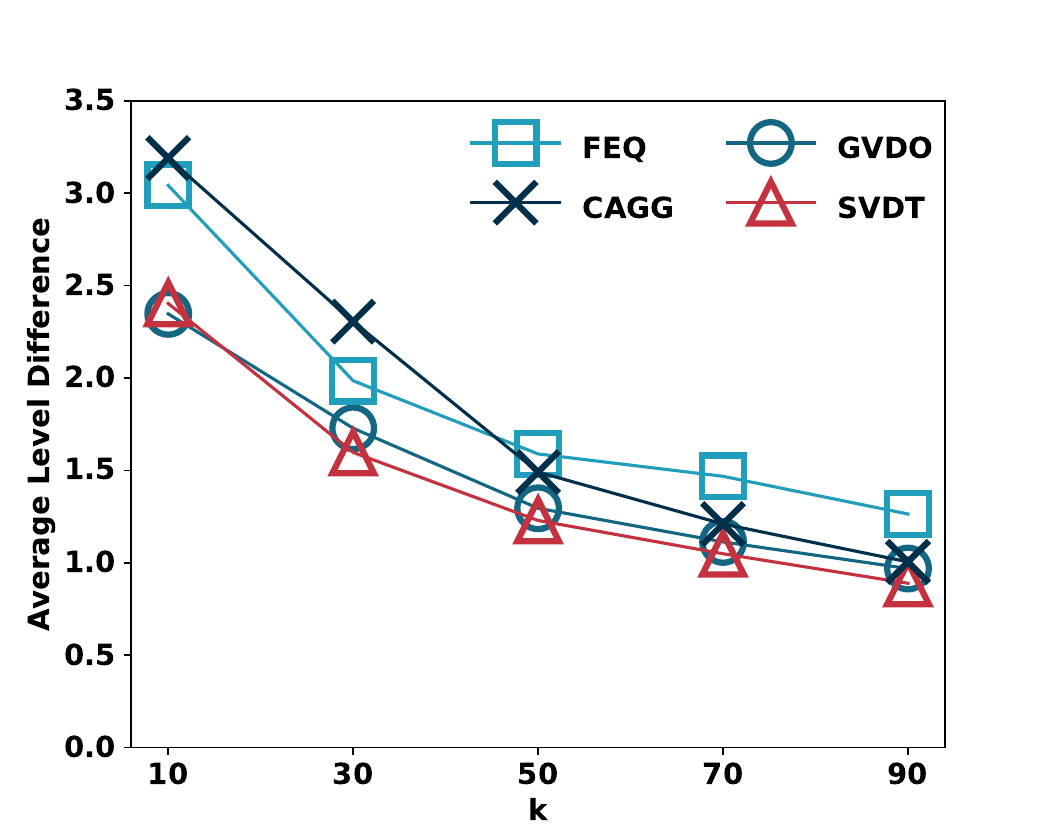} }
\hspace{-1.8em}
\subfigure[Anime2018$\&$2024]{
\includegraphics[width=1.25in]{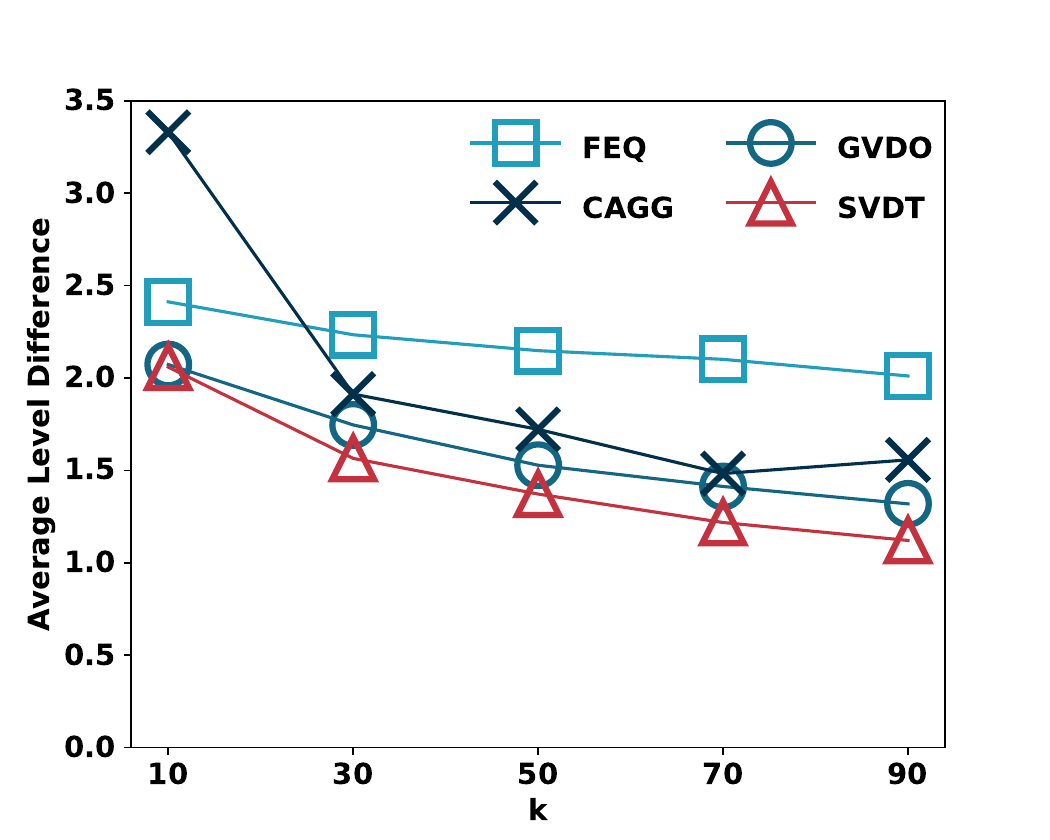}}
\hspace{-1.8em}
\subfigure[Yago1$\&$2]{
\includegraphics[width=1.25in]{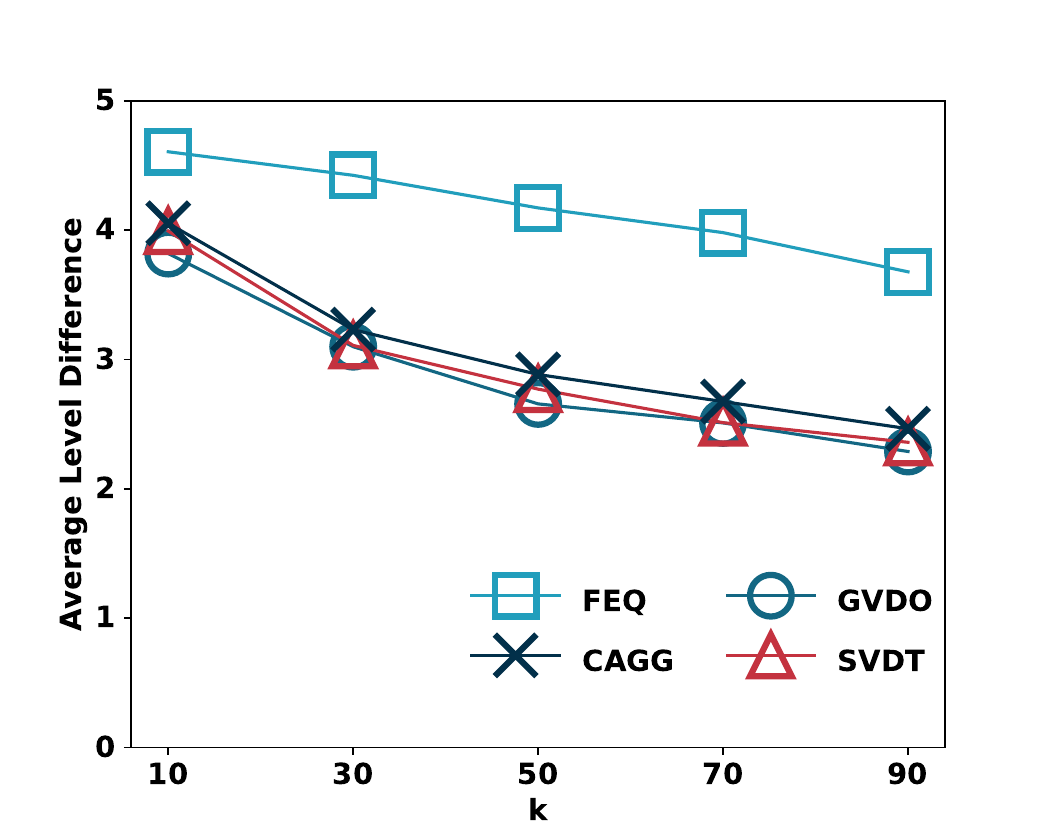}}
\hspace{-2.5em}
\vspace{-5pt}
\caption{Average level difference $Ald(S)$ evaluation.}
\label{Fig.main}
\vspace{-15pt}
\end{figure}

\begin{figure}[!h]
\centering

\includegraphics[width=1.6in]{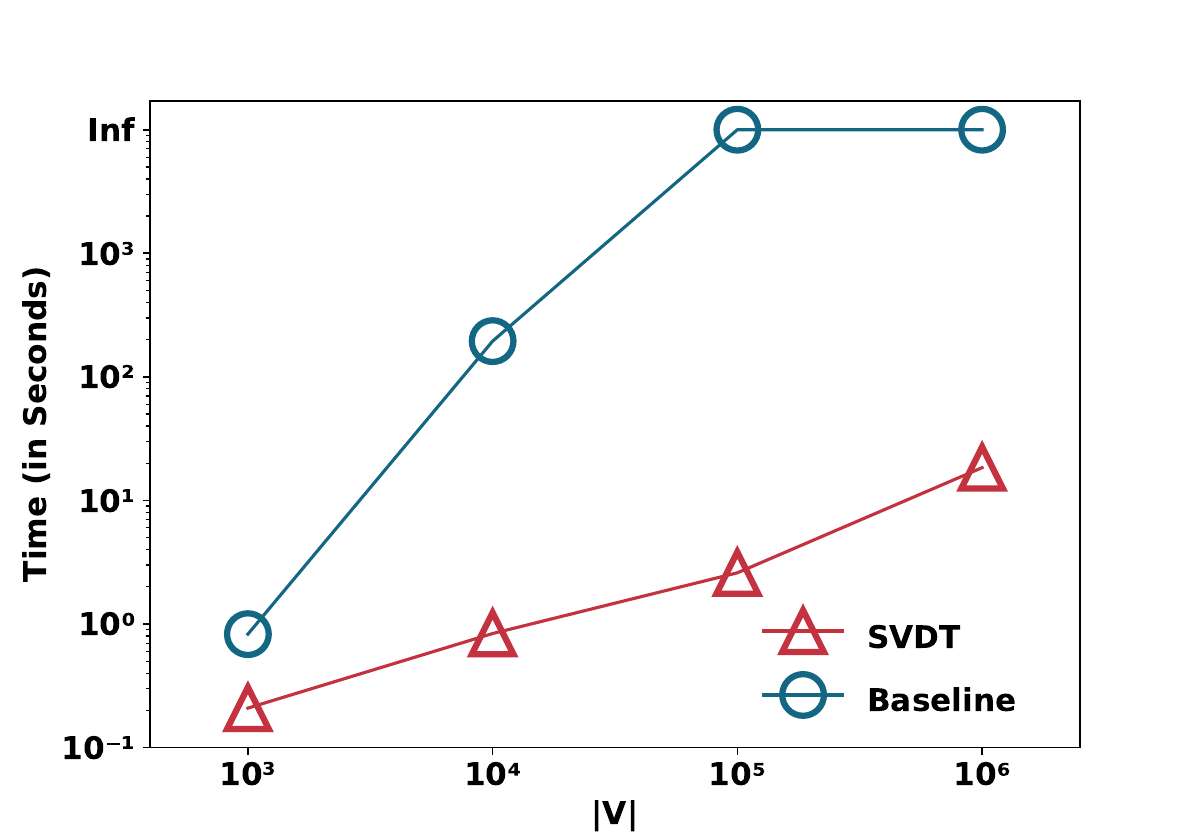}
\vspace{-5pt}
\caption{Scalability test on synthetic datasets.} \label{fig1}
\centering
\vspace{-15pt}
\end{figure}

\stitle{Exp-\MakeUppercase{\romannumeral 2}: Scalability Test}.  We generated four synthetic datasets with tree node sizes from $10^3$ to $10^6$. In all experiments, we fixed $k = 10$ and $\beta = 3$ to compare the scalability of SVDT, 
and the baseline in Fig.~8 is the optimized summary combination algorithm. 
As the tree node size increases, SVDT exhibits remarkable scalability and outperforms the baseline algorithm in efficiency, reflecting the effectiveness of Algorithm 2.

\stitle{Exp-\MakeUppercase{\romannumeral 3}: Case Study on ACM CCS Topic Summarization}. Fig.~9(a) illustrates the publication counts of research papers in different five-year periods on a poly-hierarchical ontology (ACM Computing Classification System \cite{acmccs}). 
The queries for different topics ($I_0$, $I_1$, ..., $I_{16}$) in blue, consist of specific query term names indicated in black and the corresponding number of research paper publications shown in red. For the case $k = 7$, in Fig.~9(b), we show the summarization given by the GVDO with the best experimental results among all baseline approaches in quality evaluation. It can be seen that GVDO summarizes some attractive topics 
without indicating priority or changes, thus providing limited information to users. The tree summarization by our method SVDT can be seen in Fig.~9(c). Our summarization has red and blue themes, representing similarities and differences, respectively. 
For the four second-level topics ($I_1$, $I_2$, $I_3$, $I_4$) in the original graph,  SVDT uses varying shades of red to indicate the popularity levels of these topics. Additionally, we show the diversity by highlighting the changes in some topics between the two five-year periods. Nodes ($I_5$, $I_{13}$, $I_{14}$) are colored with different shades of blue, and their direct ancestor computing methodology is marked with the deepest shade of red. These show that the popularity of artificial intelligence, learning paradigms, and machine learning methods has changed significantly in five years and is showing a trend of increasing popularity.



\begin{figure}[t]
\centering
\includegraphics[width=3.8in]{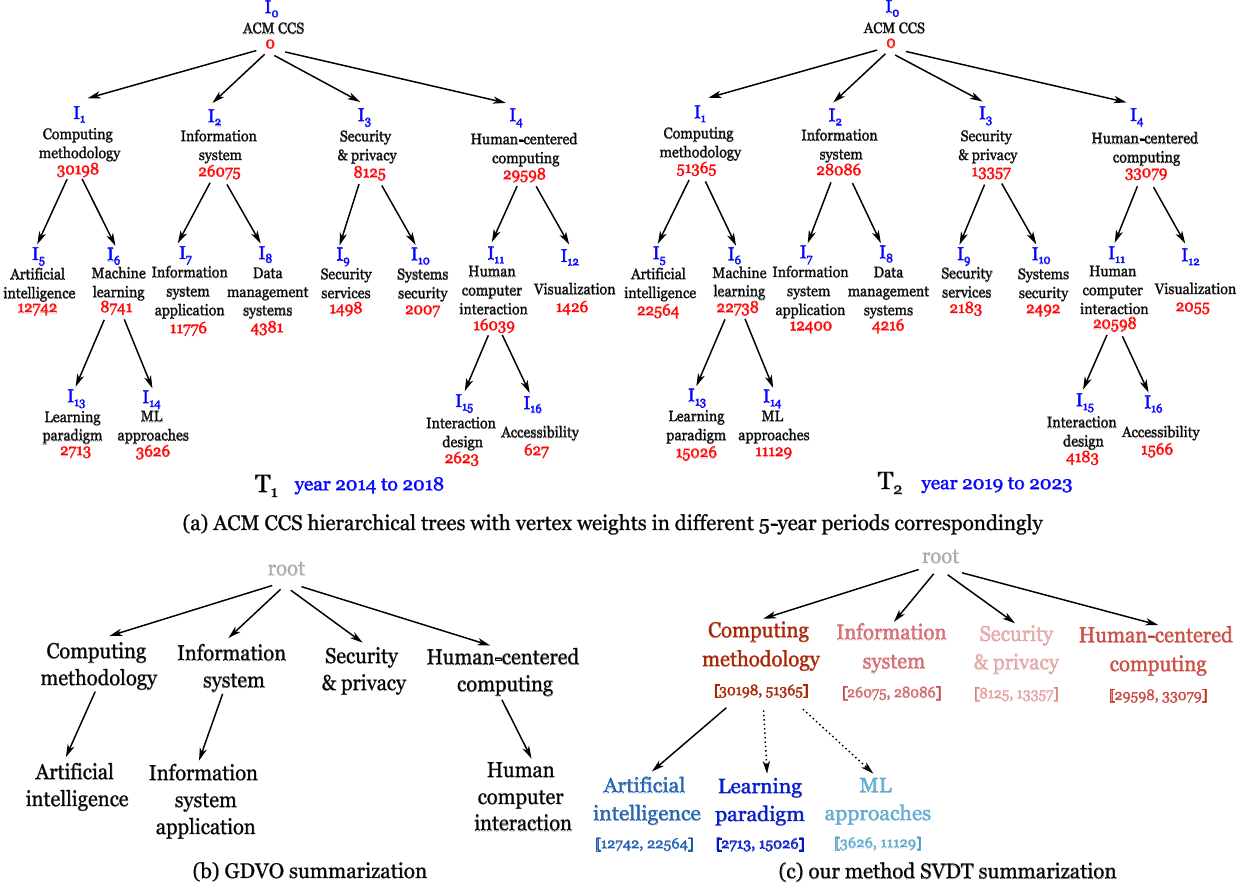}

\caption{Case study on ACM CCS hierarchical trees.} \label{fig1}
\vspace{-20pt}
\centering
\end{figure}

\vspace{-5pt}
\section{Conclusion} In this paper, we introduce the kVDT problem aimed at visually summarizing two weighted hierarchical trees of identical structure. The SVDT algorithm selects $k$ representative nodes, summarizing the commonalities and diversities between two weighted trees. We also propose a distribution normalization technique extending Hellinger distance for efficient non-leaf node computation, capturing key similarities and differences. Our method includes a summary visual representation for the compact visualization. Extensive experiments demonstrate the superiority of our proposed algorithm. 





\end{document}